\documentclass[11pt]{article}

\author{Vadim Lozin}

\date{}
\title{Graph parameters, Ramsey theory \\ and the speed of hereditary properties\thanks{Some results presented in this paper appeared in the extended abstract 
\cite{IWOCA2017} published in the proceedings of the 28th International Workshop on Combinatorial Algorithms, IWOCA 2017.}} 

\usepackage{graphicx}
\usepackage{tabularx}
\usepackage{url}
\usepackage{amsthm} 
\usepackage{amsmath} 
\usepackage{amsfonts}
\usepackage{amssymb}
\usepackage{tikz}
\usetikzlibrary{decorations.pathreplacing}

\tikzstyle{vertex}=[circle,fill=black!100,text=white,inner sep=0.8mm]
\tikzstyle{point}=[circle,fill=black,inner sep=0.1mm]

\theoremstyle{plain}
\newtheorem{theorem}{Theorem}
\newtheorem{lemma}{Lemma}

\newtheorem{corollary}{Corollary}
\theoremstyle{definition}
\newtheorem{definition}{Definition}

\theoremstyle{remark}

\begin{document}
\maketitle

\begin{abstract}
The speed of a hereditary property $P$ is the number $P_n$ of $n$-vertex labelled graphs in $P$. 
It is known that the rates of growth of $P_n$ constitute discrete layers and the speed jumps, 
in particular, from constant to polynomial, from polynomial to exponential and from exponential to 
factorial. One more jump occurs when the entropy $\lim_{n\to\infty}\frac{\log_2 P_n}{\binom{n}{2}}$ changes from 0 to a nonzero value. 
In the present paper, for each of these jumps we identify a graph parameter responsible for it,
i.e. we show that a jump of the speed coincides with a jump of the respective parameter from finitude to infinity. 
In particular, we show that 
\begin{itemize}
\item[--] the speed of a hereditary property $P$ is sub-factorial if and only if the neighbourhood diversity of graphs in $P$ is bounded by a constant,  
\item[--] the entropy of a hereditary property $P$ is 0 if and only if the VC-dimension of graphs in $P$ is bounded by a constant.  
\end{itemize}
All the result are obtained by Ramsey-type arguments.  
\end{abstract}


\section{Introduction}


A {\it graph property} is an infinite class of graphs closed under isomorphism. 
A property is {\it hereditary} if it is closed under taking induced subgraphs.
The number of $n$-vertex labelled graphs in a property $P$ is known as 
the {\it speed} of $P$ and is denoted by $P_n$.

According to Ramsey's Theorem, there exist precisely two minimal hereditary properties: the complete graphs and the edgeless graphs. 
In both cases, the speed is obviously $P_n=1$. On the other extreme, lies the set of all simple graphs, in which case the speed is $P_n=2^{\binom{n}{2}}$.
Between these two extremes, there are uncountably many other hereditary properties and their speeds have been extensively studied, 
originally in the special case of a single forbidden subgraph, and more recently in general. For example,
Erd\H os et al. \cite{EKR76} and Kolaitis et al. \cite{KPR87} studied $K_r$-free graphs, Erd{\H{o}}s et al. \cite{EFR86} studied properties where
a single graph is forbidden as a subgraph (not necessarily induced), and Pr\"omel and Steger obtained a
number of results \cite{PS1,PS2,PS3} for properties defined by a single forbidden induced subgraph. This line
of research culminated in a breakthrough result stating that for every hereditary property $P$ different from
the set of all finite graphs, the entropy $\lim_{n\to\infty}\frac{\log_2 P_n}{\binom{n}{2}}$ satisfies
\begin{equation}
\lim_{n\to\infty}\frac{\log_2 P_n}{\binom{n}{2}}=1-\frac{1}{k(P)},
\end{equation}
where $k(P)$ is a natural number called the {\it index} of $P$. To define this notion, let us denote by ${\cal E}_{i,j}$ the class
of graphs whose vertices can be partitioned into at most $i$ independent sets and $j$ cliques. In particular,
${\cal E}_{2,0}$ is the class of bipartite graphs, ${\cal E}_{0,2}$ is the class of co-bipartite (i.e. complements of bipartite) graphs and ${\cal E}_{1,1}$ is the class of split graphs. 
Then $k(P)$ is the largest $k$ such that $P$ contains ${\cal E}_{i,j}$ with $i+j = k$. This result was obtained independently by Alekseev \cite{Ale92} and Bollob\' as
and Thomason \cite{BT95,BT97} and is known nowadays as the Alekseev-Bollob\'as-Thomason Theorem (see e.g. \cite{almost}). 
%
This theorem shows that the set of possible values for the entropy is not continuous, but in fact undergoes a series of discrete `jumps'.
In particular, the entropy jumps from 0  to 1/2.

A systematic study of hereditary properties of low speed (0 entropy) was initiated by Scheinerman and Zito in \cite{low}. 
In particular, they revealed that the speed jumps 
\begin{itemize}
\item from constant (classes $P$ with $P_n = \Theta(1)$) to polynomial ($P_n = n^{\Theta(1)})$, 
\item from polynomial to exponential ($P_n = 2^{\Theta(n)})$ and 
\item from exponential to factorial ($P_n = n^{\Theta(n)})$.
\end{itemize} 
Independently, similar results have been obtained by Alekseev in \cite{Ale97}.
Moreover, Alekseev described the set of minimal classes in all the four lower layers and the asymptotic structure of properties in the first three of them. 

In the present paper we complement this line of research by identifying graph parameters responsible for the above mentioned jumps.
In particular, we show that 
\begin{itemize}
\item the speed of a hereditary property $P$ is sub-factorial if and only if the neighbourhood diversity of graphs in $P$ is bounded by a constant,  
\item the entropy of a hereditary property $P$ is 0 if and only if the VC-dimension of graphs in $P$ is bounded by a constant.  
\end{itemize}

The organization of the paper is as follows. We start by mentioning Ramsey's Theorem and related results pertinent to the topic of the paper in Section~\ref{sec:Ramsey}.
Then  we introduce a number of graph parameters and prove some results about them in Section~\ref{sec:parameters}.
Finally, we apply the results of Section~\ref{sec:parameters} to the speed of hereditary properties in Section~\ref{sec:speed}. 
Section~\ref{sec:con} concludes the paper with a number of open problems.  
In the rest of the present section, we introduce basic definitions and notations used in the paper.

\medskip
We consider only simple undirected graphs without loops and multiple edges and denote the vertex set and the edge set of a graph $G$ by $V(G)$ and $E(G)$, respectively.
If $v$ is a vertex of $G$, then $N(v)$ is its {\it neighbourhood}, i.e. the set of vertices of $G$ adjacent to $v$. The {\it closed neighbourhood} of $v$ is  defined and is denoted as $N[v]=N(v)\cup \{v\}$.
The {\it degree} of $v$ is $|N(v)|$.
In a graph, 
\begin{itemize}
\item an {\it independent set} is a subset of vertices no two of which are adjacent, 
\item a {\it clique} is a subset of vertices every two of which are adjacent, and 
\item a {\it matching} is a subset of edges no two of which share a vertex. 
\end{itemize}
For a graph $G$, we denote by $\overline{G}$ the complement of $G$. Similarly, for a class $X$ of graphs, we denote by $\overline{X}$ the class of complements of graphs in $X$.

Given a graph $G$ and a subset $U\subseteq V(G)$, we denote by $G[U]$ the subgraph of $G$ induced by $U$, i.e. the subgraph obtained from $G$ by
deleting all the vertices not in $U$. We say that a graph $G$ contains a graph $H$ as an induced subgraph if $H$ is isomorphic to an induced subgraph of $G$.
A graph $G$ is said to be $n$-universal for a class of graphs $X$ if $G$ contains all $n$-vertex graphs from $X$ as induced subgraphs. 

A class $X$ of graphs is {\it hereditary} if it is closed under taking induced subgraphs, 
i.e. if $G\in X$ implies $H\in X$ for every graph $H$ contained in $G$ as an induced subgraph. 
Two hereditary classes of particular interest in this paper are split graphs and bipartite graphs.

A graph $G$ is a {\it split} graph if $V(G)$ can be partitioned into an independent set and a clique, and $G$ is {\it bipartite} if $V(G)$ 
can be partitioned into at most two independent sets. A bipartite graph $G$ given together with 
a bipartition of its vertices into independent sets $A$ and $B$ will be denoted $G=(A,B,E)$, in which case we will say that $A$ and $B$ are 
the color classes or simply parts of $G$. If every vertex of $A$ is adjacent to every vertex of $B$, then $G=(A,B,E)$ is {\it complete bipartite}, also known as a {\it biclique}. 
A {\it star} is a complete bipartite graph with one of the color classes being of size 1. The star with the second color class being of size $n$ is denoted $K_{1,n}$.
The {\it bipartite complement} of a bipartite graph  $G=(A,B,E)$ is the bipartite graph $G'=(A,B,E')$, where $ab\in E'$ if and only if $ab\not\in E$. 
Clearly, by creating a clique in one of the color classes of a bipartite graph, we transform it into a split graph, and vice versa.  

Given a graph $G$ and two disjoint subset $A,B$ of its vertices, we denote by $G[A,B]$ the bipartite (not necessarily induced) subgraph with color classes $A$ and $B$ formed by the edges of $G$
between $A$ and $B$. 

\section{Ramsey theory}
\label{sec:Ramsey}

In 1930, a 26 years old British mathematician Frank Ramsey proved the following theorem, known nowadays as Ramsey's Theorem. 

\begin{theorem}{\rm \cite{Ramsey}}\label{thm:Ramsey}
For any positive integers $k$, $r$, $p$, there exists a minimum positive integer $F=F(k,r,p)$ such that 
if the $k$-subsets of an $F$-set are colored with $r$ colors, then there is a monochromatic $p$-set, i.e.
a $p$-set all of whose $k$-subsets have the same color.  
\end{theorem}

We will refer to the number $F(k,r,p)$ defined in this theorem as the Frank Ramsey number.

It is not difficult to see that with $k=1$ the theorem coincides with the Pigeonhole Principle.
For $k=2$, the theorem admits a nice interpretation in the terminology of graph theory, since 
coloring 2-subsets can be viewed as coloring the edges of a complete graph: 
for any positive integers $r$ and $p$, there is a positive integer $n=n(r,p)$ such that 
if the edges of an $n$-vertex complete graph are colored with $r$ colors, then there is a monochromatic clique of
size $p$, i.e. a clique all of whose edges have the same color.  

In the case of $r=2$ colors, the graph-theoretic interpretation of Ramsey's Theorem can be further rephrased as follows.

\begin{theorem}\label{thm:Ramsey-graph}
For any  positive integer $p$, there is a minimum positive integer $R(p)$ such that every graph
with at least $R(p)$ vertices has either a clique of size $p$ or an independent set of size $p$.
\end{theorem}

The number $R(p)$  is known as
the {\it symmetric Ramsey number}. In other words, $R(p)=F(2,2,p)$, where $F(k,r,p)$ is the Frank Ramsey number.
Theorem~\ref{thm:Ramsey-graph} also admits a non-symmetric formulation as follows. 

\begin{theorem}\label{thm:Ramsey-graph-n}
For any  positive integers $p$ and $q$, there is a minimum positive integer $R(p,q)$ such that every graph
with at least $R(p,q)$ vertices has either a clique of size $p$ or an independent set of size $q$.
\end{theorem}

The number $R(p,q)$ is known as  the {\it Ramsey number} $R(p,q)$. In particular, $R(p,p)=R(p)$.

Theorem~\ref{thm:Ramsey-graph-n} (or Theorem~\ref{thm:Ramsey-graph}) allows us to make the following conclusion: if $X$ is a hereditary class
which does not contain a complete graph $K_p$ and an edgeless graph $\overline{K}_q$, then graphs in $X$ have fewer than $R(p,q)$ vertices, i.e. $X$ is finite.
More formally, Theorem~\ref{thm:Ramsey-graph-n} implies the following conclusion.

\begin{theorem}\label{thm:minimal}
The class of complete graphs and the class of edgeless graphs are the only two minimal infinite hereditary classes of graphs. 
\end{theorem}

On the other hand, it is not difficult to see that the reverse is also true: Theorem~\ref{thm:minimal} implies Theorem~\ref{thm:Ramsey-graph-n}.
In other words, these two theorems are equivalent.

\medskip
Theorem~\ref{thm:minimal} characterizes the family of hereditary classes containing graphs with a bounded number of vertices 
in terms of minimal ``forbidden'' elements, i.e. minimal classes where the vertex number is unbounded.  
It turns out that various other parameters admit a similar characterization. For instance, directly from Ramsey's Theorem it follows  that 
\begin{itemize}
\item the class of complete graphs and the class of stars (and all their induced subgraphs) are the only two minimal hereditary classes of graphs of unbounded vertex degree.
\item the class of complete graphs and the class of complete bipartite graphs are the only two minimal hereditary classes of graphs of unbounded
biclique number, where the biclique number of a graph $G$ is the size a maximum complete bipartite (not necessarily induced) subgraph of $G$ with equal parts.  
\end{itemize}

In the next section, we use Ramsey-type arguments in order to characterize several other graph parameters by means of minimal hereditary classes,
where these parameters are unbounded. In addition to the original Ramsey's Theorem we will need its bipartite version, which can be stated as follows. 

\begin{theorem}\label{thm:Ramsey-bipartite}
For every $s$, there is an $n = n(s)$ such that every bipartite graph $G$ with at
least $n$ vertices in each part contains either a biclique with color classes of size $s$ or its bipartite complement. 
\end{theorem}


\section{Graph parameters}
\label{sec:parameters}


\subsection{Independence number, clique number and complex number}


As usual, $\alpha(G)$ stands for the independence number of $G$, i.e. the size of a maximum independent set in $G$, and $\omega(G)$ for the clique number of $G$, i.e. 
the size of a maximum clique in $G$.
Now let us introduce a new parameter:
\begin{itemize}
\item[$c(G)$]= $\min(\alpha(G),\omega(G))$ is the {\it complex number} of $G$. 
\end{itemize}   
  
In what follows we give a Ramsey-type characterization of this parameter, i.e. we characterize it in terms of 
minimal hereditary classes where the parameter is unbounded.  To this end, let us denote by

\begin{itemize}
\item[$\cal S$] the class of graphs partitionable into a clique and a set of isolated vertices. 
Also, let $S_{n}$ be a graph in $\cal S$ with a clique of size $n$ and a set of isolated vertices of size $n$. Obviously, 
$S_{n}$ is $n$-universal for graphs in $\cal S$, i.e. it contains every $n$-vertex graph from $\cal S$ as an induced subgraph. 
\end{itemize}

\begin{theorem}\label{thm:complex}
$\cal S$ and $\overline{\cal S}$ are the only two minimal hereditary classes of graphs of unbounded complex number.
\end{theorem}  

\begin{proof}
Obviously, the complex number of graphs in $\cal S$ and $\overline{\cal S}$ can be arbitrarily large. 
Conversely, let $X$ be a hereditary class containing a graph $G$ with $c(G)\ge k$ for each value of $k$.
Then $G$ contains a clique $C$ of size $k$ and an independent set $I$ of size $k$. Since $C$ and $I$ 
have at most one vertex in common, we may assume without loss of generality that they are disjoint,
and since $k$ can be arbitrarily large, in the bipartite graph $G[C,I]$
we can find an arbitrarily large biclique or its bipartite complement (Theorem~\ref{thm:Ramsey-bipartite}).  
Therefore, graphs in $X$ contain either $S_{n}$ or $\overline{S}_{n}$  for arbitrarily large values of $n$,
i.e. $X$ contains either $\cal S$ and $\overline{\cal S}$.
\end{proof}  


\subsection{Degree, co-degree and complex degree}
\label{sec:degree}
Let $G$ be a graph and $v$ a vertex of $G$. We denote by $d(v)$ the degree of $v$ and by $\overline{d}(v)$ the {\it co-degree} of $v$, 
i.e. the degree of $v$ in the complement of $G$. The $c$-degree of $v$ is denoted and defined as follows: $cd(v)=\min(d(v),\overline{d}(v))$. 

As usual, $\Delta(G)$ is the maximum vertex degree in  $G$. Also, we denote by $\overline{\Delta}(G)$ the maximum co-degree and 
by $c\Delta(G)$ the maximum $c$-degree in $G$. We call $c\Delta(G)$ the {\it complex degree} in $G$. 
In order to characterize this new parameter  by means of minimal hereditary classes 
where  complex degree is unbounded, let us denote by 

\begin{itemize}
\item[$\cal Q$] the class of graphs whose vertices can be partitioned into a set inducing a star and a set of isolated vertices.
Also, let $Q_{n}$ be a graph in $\cal Q$ whose vertices can be partitioned into an induced star $K_{1,n}$ 
and a set of isolated vertices of size $n$. 
Obviously, $Q_{n}$ is $n$-universal for graphs in $\cal Q$, i.e. it contains every $n$-vertex graph from $\cal Q$ as an induced subgraph. 
\item[$\cal B$] the class of complete bipartite graphs (an edgeless graph is counted as complete bipartite with one part being empty). 
For consistency of notation with previously defined classes, we will denote a biclique (complete bipartite graph) 
with $n$ vertices in each part of its bipartition by $B_n$. 
Clearly, $B_{n}$ is $n$-universal for graphs in $\cal B$, i.e. it contains all $n$-vertex graphs in $\cal B$ as induced subgraphs.  
\end{itemize}

\begin{theorem}\label{thm:complex-degree}
$\cal S$, $\overline{\cal S}$, $\cal Q$, $\overline{\cal Q}$, $\cal B$ and $\overline{\cal B}$ are the only minimal hereditary classes of graphs of unbounded complex degree.
\end{theorem}  

\begin{proof}
Obviously, the complex degree of graphs in $\cal S$, $\overline{\cal S}$ $\cal Q$, $\overline{\cal Q}$, $\cal B$ and $\overline{\cal B}$ can be arbitrarily large. 
Conversely, let $X$ be a hereditary class containing a graph $G$ with $c\Delta(G)\ge k$ for each value of $k$.
Then $G$ contains a vertex $v$ such that $d(v)\ge k$ and $\overline{d}(v)\ge k$. 
Let $A$ be the set of neighbours of $v$ and $B$ the set of its non-neighbours. Since 
both $A$ and $B$ can be arbitrarily large, each of them contains either a big clique or a big independent set (Theorem~\ref{thm:Ramsey-graph-n}),
and $G[A,B]$ contains either a big biclique or its bipartite complement (Theorem~\ref{thm:Ramsey-bipartite}).
Therefore, graphs in $X$ contain either $S_{n}$ or $\overline{S}_{n}$ or $Q_{n}$ or $\overline{Q}_{n}$ or $B_n$ or $\overline{B}_n$  for arbitrarily large values of $n$.
As a result, $X$ contains at least one of $\cal S$, $\overline{\cal S}$, $\cal Q$, $\overline{\cal Q}$, $\cal B$ and $\overline{\cal B}$.
\end{proof}


\subsection{Matching number, co-matching number and $c$-matching number}
\label{sec:matching}

The {\it matching number} of a graph $G$ is the size of a maximum matching in $G$ and we denote it by $\mu(G)$. 
The {\it co-matching number} of $G$ is the size of a maximum matching in the complement of $G$ and we denote it by $\overline{\mu}(G)$. 
The {\it $c$-matching number} of $G$ is defined and denoted as follows: $c\mu(G)=\min(\mu(G),\overline{\mu}(G))$.
In this section, we characterize all three parameters in terms of minimal hereditary classes where these parameters are unbounded.
To this end, let us denote by  

\begin{itemize}
\item[$\cal M$] the class of graphs of vertex degree at most 1.
By $M_n$ we denote an induced matching of size $n$, i.e. the unique up to isomorphism graph from this class with $2n$ vertices each of which has degree 1. 
Clearly, $M_{n}$ is $n$-universal for graphs in $\cal M$.
\item[$\cal Z$] the class of chain graphs. These are bipartite graphs in which the vertices in each color class can be linearly ordered under inclusion of their neighbourhoods, i.e. the neighbourhoods form a chain. 
By $Z_n$ we denote a chain graph such that for each $i\in \{1,2,\ldots,n\}$,  each part of the graph contains exactly one vertex of degree $i$.
Figure~\ref{fig:B5} represents the graph $Z_n$ for $n=5$. It is known \cite{chain-uni} that $Z_n$ is $n$-universal for graphs in $\cal Z$.
\end{itemize}

\begin{figure}[ht]
\begin{center}
\begin{picture}(300,60)
\put(50,0){\circle*{5}}
\put(100,0){\circle*{5}}
\put(150,0){\circle*{5}}
\put(200,0){\circle*{5}}
\put(250,0){\circle*{5}}
\put(48,53){$y_1$}
\put(98,53){$y_2$}
\put(148,53){$y_3$}
\put(198,53){$y_4$}
\put(248,53){$y_5$}



\put(48,-12){$x_1$}
\put(98,-12){$x_2$}
\put(148,-12){$x_3$}
\put(198,-12){$x_4$}
\put(248,-12){$x_5$}
\put(50,50){\circle*{5}}
\put(100,50){\circle*{5}}
\put(150,50){\circle*{5}}
\put(200,50){\circle*{5}}
\put(250,50){\circle*{5}}

\put(50,0){\line(0,1){50}}
\put(250,0){\line(-1,1){50}}
\put(250,0){\line(-2,1){100}}
\put(250,0){\line(-3,1){150}}
\put(250,0){\line(-4,1){200}}
\put(100,0){\line(0,1){50}}
\put(200,0){\line(-1,1){50}}
\put(200,0){\line(-2,1){100}}
\put(200,0){\line(-3,1){150}}
\put(150,0){\line(0,1){50}}
\put(150,0){\line(-1,1){50}}
\put(150,0){\line(-2,1){100}}
\put(200,0){\line(0,1){50}}
\put(100,0){\line(-1,1){50}}
\put(250,0){\line(0,1){50}}

\end{picture}
\end{center}
\caption{The graph $Z_5$}
\label{fig:B5}
\end{figure}
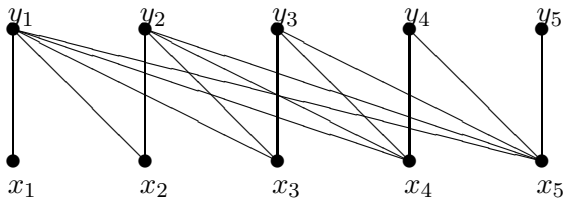

\begin{lemma}\label{lem:matching}
For any positive integers $s,t$, there exists a positive integer $q=q(s,t)$ such that every bipartite graph $G$ with a matching of size $q$ 
contains either an induced $M_s$ or an induced $B_{t}$.  
\end{lemma}

\begin{proof}
Let us denote $m=2\max(s,t)$ and $q=F(2,4,m)$, where $F(k,r,p)$ is the Frank Ramsey number (Theorem~\ref{thm:Ramsey}).
Consider a matching $M=\{x_1y_1,\ldots,x_qy_q\}$ of size $q$. 
We color each pair $x_iy_i,x_jy_j$ $(i<j)$ of edges of $M$  in one of the four colors as follows:
\begin{itemize}
\item color 1 if $G$ contains no edges between $x_iy_i$ and $x_jy_j$,
\item color 2  if $G$ contains both possible edges between $x_iy_i$ and $x_jy_j$,
\item color 3 if $G$ contains the edge $x_iy_j$ but not the edge $y_ix_j$,
\item color 4 if $G$ contains the edge $y_ix_j$ but not the edge $x_iy_j$.
\end{itemize} 
By Ramsey's Theorem, $M$ contains a monochromatic set $M'$ of  edges of size $m$. If the color of each pair in $M'$ is
\begin{itemize}
\item[1,] then $M'$ is an induced matching of size $m\ge 2s>s$, 
\item[2,] then the vertices of $M'$ induce a complete bipartite graph $B_{m}$ with $m\ge 2t>t$, 
\item[3] or 4, then the vertices of $M'$ induce a $Z_{m}$ and hence $G$ contains a complete bipartite graph $B_{m/2}$ with $m/2\ge t$. 
\end{itemize}~\end{proof}

\begin{lemma}\label{lem:m}
For any natural $s,t,p$, there exists a $Q=Q(s,t,p)$ such that every graph $G$ with a matching of size $Q$ 
contains either an induced $M_s$ or an induced $B_{t}$ or a clique $K_p$.  
\end{lemma}

\begin{proof}
Let us denote $Q=R(p,R(p,q))$, where $R$ is the (non-symmetric) Ramsey number and  
$q=F(2,4,2\max(s,t))$ is  the value defined in the proof of Lemma~\ref{lem:matching}.
We consider a matching $M$ of size $Q$ in $G$ and color the endpoints of each edge of $M$ in two colors, say white and black, arbitrarily. 
Since the set of white vertices has size $Q$, it must contain either a clique $K_p$, in which case we are done, or an independent set $A$ of size $R(p,q)$.
In the latter case, we look at the black vertices matched with the vertices of $A$. According to the size of this set, 
it must contain either a clique $K_p$, in which case we are done, or an independent set $A'$ of size $q$. In the latter case, we denote by $A''$ the set
of white vertices matched with the vertices of $A'$. Then $A'$ and $A''$ induce a bipartite graph with a matching of size $q$,
in which case, by Lemma~\ref{lem:matching}, $G$ contains either an induced matching of size $s$ or an induced complete bipartite graph $B_{t}$.  
\end{proof}

\medskip
The above sequence of results allows us to make the following conclusions, the first two of which follow directly from Lemma~\ref{lem:m}.

\begin{theorem}
$\cal M$, $\cal B$ and the class of complete graphs are the only three minimal hereditary classes of graphs of unbounded matching number.  
\end{theorem}

\begin{theorem}
$\overline{\cal M}$, $\overline{\cal B}$ and the class of edgeless graphs are 
the only three minimal hereditary classes of graphs of unbounded co-matching number.  
\end{theorem}

\begin{theorem}\label{thm:c-matching}
$\cal M$, $\cal B$,  $\cal S$, $\overline{\cal M}$, $\overline{\cal B}$ and $\overline{\cal S}$ are 
the only six minimal hereditary classes of graphs of unbounded $c$-matching number.  
\end{theorem}

\begin{proof}
Clearly, graphs in $\cal M$, $\cal B$,  $\cal S$, $\overline{\cal M}$, $\overline{\cal B}$ and $\overline{\cal S}$ 
can have arbitrarily large $c$-matching number.  
Conversely, let $X$ be a hereditary class with unbounded $c$-matching number. Assume $X$ contains none of 
$\cal M$, $\cal B$,  $\overline{\cal M}$, $\overline{\cal B}$,
i.e. there is a value of $p$ such that none of $M_p$, $B_{p}$, $\overline{M}_p$, $\overline{B}_{p}$ belongs to $X$. 
By assumption, $X$ contains a graph $G$ with $c\mu(G)\ge k$ for each value of $k$, i.e. $G$ contains a matching and a co-matching of size $k$.
Since $k$ can be arbitrarily large and $M_p$, $B_{p}$ are forbidden, $G$ contains a large clique (Lemma~\ref{lem:m}). 
Similarly, $G$ contains a large independent set. Therefore, $X$ contains graphs with arbitrarily large complex number. But then $X$ contains 
either $\cal S$ or $\overline{\cal S}$ (Theorem~\ref{thm:complex}).
\end{proof}


\subsection{Neighbourhood diversity}
\label{sec:diversity}


\begin{definition}
Let us say that two vertices $x$ and $y$ are {\it similar} if there is no vertex $z$ distinguishing them
(i.e. if there is no vertex $z$ adjacent to exactly one of $x$ and $y$). Clearly, the similarity is an equivalence relation. We denote by 
$nd(G)$ the number of similarity class in $G$ and call it the {\it neighbourhood diversity} of $G$.
\end{definition}

In order to characterize the neighbourhood diversity by means of minimal hereditary classes of graphs where this parameter is unbounded,
we denote by 

\begin{itemize}
\item[${\cal M}^{bc}$] the class of bipartite complements of graphs in $\cal M$. The bipartite complement of the graph $M_n$ will be denoted ${M}^{bc}_n$.
Clearly, ${M}^{bc}_n$ is $n$-universal for graphs in ${\cal M}^{bc}$.
\item[${\cal M}^*$] the class of split graphs obtained from graphs in ${\cal M}$ by creating a clique in one of the color classes. The graph obtained from $M_n$ 
by  creating clique in one its color classes will be denoted by $M_n^*$. 
Clearly, $M_n^*$ is $n$-universal for graphs in ${\cal M}^*$.
\item[${\cal Z}^*$] the class of split graphs obtained from graphs in ${\cal Z}$ by creating a clique in one of the color classes. 
This class is known in the literature as the class of {\it threshold graphs}. The graph obtained from $Z_n$  by creating a clique in one of its color classes will be denoted 
$Z^*_n$.  This graph is $n$-universal for threshold graphs \cite{threshold}.
\end{itemize}

Before we provide a Ramsey-type characterization of the neighbourhood diversity, we introduce an auxiliary parameter.

\begin{definition}
A {\it skew matching} in a graph $G$ is a matching $\{x_1y_1,\ldots,x_qy_q\}$ such that $y_i$ is not adjacent to $x_j$ for all $i<j$.
The {\it complement of a skew matching} is a sequence of pairs of vertices that create a skew matching in the complement of $G$.
\end{definition}

\begin{lemma}\label{lem:skew}
For any positive integer $m$, there exists a positive integer $r=r(m)$ such that any bipartite graph $G=(A,B,E)$ of neighbourhood diversity $r$
contains either a skew matching of size $m$ or its complement.
\end{lemma}

\begin{proof}
Define $r=2^{2m}$ and let $D$ be a set of pairwise non-similar vertices of size $r/2$ chosen from the same color class of $G$, say from $A$.
Let $y_1$ be a vertex in $B$ distinguishing the set $D$ (i.e. $y_1$ has both a neighbour and a non-neighbour in $D$) and 
let us say that $y_1$ is {\it big} if the number of its neighbours in $D$ is larger than the number of its non-neighbours in $D$,
and {\it small} otherwise. If 
\begin{itemize}
\item[$y_1$] is small, we arbitrarily choose its neighbour in $D$, denote it by $x_1$ and remove all neighbours of $y_1$ from $D$.
\item[$y_1$] is big, we arbitrarily choose a non-neighbour of $y_1$ in $D$, denote it by $x_1$ and remove all non-neighbours of $y_1$ from $D$.
\end{itemize}
Observe that $y_1$ does not distinguish the vertices in the updated set $D$. 

We apply the above procedure to the set $D$ $2m-1$ times and obtain in this way a sequence of $2m-1$ pairs $x_iy_i$.
If $m$ of these pairs contain small vertices $y_i$, then these pairs create a skew matching (of size $m$).
Otherwise, there is a set of $m$ pairs containing big vertices $y_i$, in which case these pairs create the complement of a skew matching.
\end{proof}

\begin{lemma}\label{lem:nd}
For any positive integer $p$, there exists a positive integer $q=q(p)$ such that any bipartite graph $G=(A,B,E)$ of neighbourhood diversity $q$
contains either an induced $M_p$ or an induced $Z_p$  or an induced $M^{bc}_p$.
\end{lemma}

\begin{proof}
Let $m=R(p+1)$ (where $R$ is the symmetric Ramsey number) and $q=2^{2m}$. According to the proof of Lemma~\ref{lem:skew}, $G$ contains a skew matching of size $m$
or its complement. If $G$ contains a skew matching $M$, we color each pair $x_iy_i$, $x_jy_j$ ($i<j$) of edges of $M$ in two colors as follows: 
\begin{itemize}
\item color 1 if $x_i$ is not adjacent to $y_j$,
\item color 2 if $x_i$ is adjacent to $y_j$.
\end{itemize} 
By Ramsey's Theorem, $M$ contains a monochromatic set $M'$ of  edges of size $p+1$. If the color of each pair of edges in $M'$ is
\begin{itemize}
\item[1,] then $M'$ is an induced matching of size $p+1$, 
\item[2,] then the vertices of $M'$ induce a $Z_{p+1}$.  
\end{itemize} 
Analogously, in the case when $G$ contains the complement of a skew matching, we find either an induced $M^{bc}_{p+1}$ 
or an induced $Z_p$ (observe that the bipartite complement of $Z_{p+1}$ contains an induced $Z_p$).
\end{proof}

\begin{lemma}\label{lem:nd1}
For any positive integer $p$, there exists a positive integer $Q=Q(p)$ such that every graph $G$ of neighbourhood diversity $Q$  
contains one of the following nine graphs as an induced subgraph: $M_p$, $M^{bc}_{p}$, $Z_p$, $\overline{M}_p$, $\overline{M}^{bc}_p$, $\overline{Z}_p$,
$M^*_p$, $\overline{M}^{*}_p$, $Z^*_p$.  
\end{lemma}

\begin{proof}
Let $Q=R(q)$, where $q=2^{2m}$ and $m=R(R(p)+1)$ ($R$ is the symmetric Ramsey number).
We choose one vertex from each similarity class of $G$ and find in the chosen set a subset $A$ of vertices that form an independent set or a clique of size $q=2^{2m}$.
Let us call the vertices of $A$ white. We denote the remaining vertices of $G$ by $B$ and call them black.
Let $G'=G[A,B]$. By the choice of $A$, all vertices of this set
have pairwise different neighbourhoods in $G'$. Therefore, according to the proof of Lemma~\ref{lem:nd},
$G'$ contains a subgraph $G''$ inducing either $M_n$, or $M^{bc}_n$  or $Z_n$ with $n=R(p)$.
Among the $n$ black vertices of $G''$, we can find 
a subset $B'$ of vertices that form either a clique or an independent set of size $p$ in the graph $G$. 
Then $B'$ together with a subset of $A$ of size $p$ induce in $G$ one of the nine graphs listed in the statement of theorem.
\end{proof}

Since the nine graphs of Lemma~\ref{lem:nd1} are universal for their respective classes, we make the following conclusion.

\begin{theorem}\label{thm:diversity}
There exist exactly nine minimal hereditary classes of graphs of unbounded neighbourhood diversity: 
$\cal M$, ${\cal M}^{bc}$, $\cal Z$, $\overline{\cal M}$, $\overline{\cal M}^{bc}$, $\overline{\cal Z}$, ${\cal M}^*$, $\overline{\cal M}^*$, ${\cal Z}^*$. 
\end{theorem}


\subsection{VC-dimension}


A set system $(X,S)$ consists of a set $X$ and a family $S$ of subsets of $X$. 
A subset $A\subseteq X$ is {\it shattered} if for every subset $B\subseteq A$
there is a set $C\in S$ such that $B=A\cap C$. The VC-dimension of $(X,S)$
is the cardinality of a largest shattered subset of $X$.

The VC-dimension of a graph $G=(V,E)$ was defined in \cite{VC} as the VC-dimension of 
the set system $(V,S)$, where $S$ the family of closed neighbourhoods of vertices of $G$,
i.e. $S=\{N[v]\ :\ v\in V(G)\}$. Let us denote the VC-dimension of $G$ by $vc[G]$.

In this section, we characterize VC-dimension 
by means of three minimal hereditary classes where this parameter is unbounded.
To this end, we first redefine it in terms of open neighbourhoods as follows.
Let $vc(G)$ be the size of a largest set $A$ of vertices of $G$ such that for any 
subset $B\subseteq A$ there is a vertex $v$ {\it outside of} $A$ with $B=A\cap N(v)$. 
In other words, $vc(G)$ is the size of a largest subset of vertices shattered by {\it open} 
neighbourhoods of vertices of $G$.

We start by showing that the two definitions are equivalent in the sense that they both are either bounded or unbounded in a hereditary class.
To prove this, we introduce the following terminology. 
Let $A$ be a set of vertices which is shattered by a collection of neighbourhoods (open or closed). 
For a subset $B\subseteq A$ we will denote by $v(B)$ the vertex whose neighbourhood (open or closed) intersect $A$ at $B$.
We will say that $B$ is {\it closed} if $v(B)$ belongs to $B$, and {\it open} otherwise. 

\begin{lemma}
$vc(G)\le vc[G]\le vc(G)(vc(G)+1)+1$
\end{lemma}

\begin{proof}
The first inequality is obvious. To prove the second one, let $A$ be a subset of $V(G)$
of size $vc[G]$ which is shattered by a collection of closed neighbourhoods. 
If $A$ has no closed subsets, then $vc[G]=vc(G)$. Otherwise, let $B$ be a closed subset of $A$. 

Assume first that $|B|=1$. Then $B=\{v(B)\}$ and $v(B)$ is isolated in $G[A]$, i.e. it has no neighbours 
in $A$. Let $C$ be the set of all such vertices, i.e. the set of vertices each of which is a closed subset of $A$.
By removing from $A$ any vertex $x\in C$ we obtain a new set $A$ and may assume that it has no closed subsets of size $1$.
Indeed, for any vertex $y\in C$ different from $x$, there must exist a vertex $y'\not\in A$ such that $N(y')\cap A=\{x,y\}$ (since $A$ is shattered).
After the removal of $x$ from $A$, we have $N(y')\cap A=\{y\}$ and hence $\{y\}$ is not a closed subset anymore. 
This discussion allows us to assume in what follows that $A$ has no closed subsets of size $1$, in which case we only need to 
show that $vc[G]\le vc(G)(vc(G)+1)$.

Assume now that $B$ is a closed subset of $A$ of size at least 2. Suppose that $B-v(B)$ contains a closed subset $C$,
i.e. $v(C)\in C$. Observe that $v(C)$ is adjacent to $v(B)$, as every vertex of $B-v(B)$ is adjacent to $v(B)$.
But then $N[v(C)]\cap A$ contains $v(B)$ contradicting the fact that $N[v(C)]\cap A=C$. This contradiction shows that every 
subset of $B-v(B)$ is open, i.e.  $|B-v(B)|\le vc(G)$.

The above observation allows us to apply the following procedure: as long as $A$ contains a closed subset $B$ with at least two vertices, 
delete from $A$ all vertices of $B$ except for $v(B)$. Denote the resulting set by $A^*$. Assume the procedure was applied $p$ times and let $B_1,\ldots,B_p$ be the closed subsets
it was applied to. It is not difficult to see that the set $\{v(B_1),\ldots,v(B_p)\}$ has no closed subsets and hence its size cannot be large than 
$vc(G)$, i.e. $p\le vc(G)$. Combining, we conclude:
$$
vc[G]=|A|\le |A^*|+\sum\limits_{i=1}^{p} |B_i-v(B_i)|\le vc(G)+p\cdot vc(G)\le vc(G)(vc(G)+1).
$$
\end{proof}

\medskip
This lemma allows us to assume that if $A$ is shattered, then there is a set $C$ {\it disjoint} from $A$ such that 
for any subset $B\subseteq A$ there is a vertex $v\in C$ with $B=A\cap N(v)$, in which case we will say that $A$ is shattered by $C$, or $C$ shatters $A$. 

\medskip
Let $W_n=(A,B,E)$ be the bipartite graph with $|A|=n$ and $|B|=2^n$ such that all vertices of $B$ have pairwise different neighbourhoods in $A$.
Also, let $D_n$ be the split graph obtained from $W_n$ by creating a clique in $A$. 

\begin{lemma}\label{lem:universal}
The graph $W_n$ is an $n$-universal bipartite graph, i.e. it contains every bipartite graph with $n$ vertices as an induced subgraph.
\end{lemma}
\begin{proof}
Let $G$ be a bipartite graph with $n$ vertices and with parts $A$ and $B$ of size $n_1$ and $n_2$, respectively. By adding at most $n_2$ vertices to $A$, 
we can guarantee that all vertices of $B$ have pairwise different neighbourhoods in $A$. 
Clearly, $W_n$ contains the extended graph and hence it also contains $G$ as an induced subgraph.    
\end{proof}

\begin{corollary}
Every co-bipartite graph with at most $n$ vertices is contained in $\overline{W}_n$ and every split graph with at most $n$ vertices is contained 
in both $D_n$ and in $\overline{D}_n$.
\end{corollary}

\begin{lemma}\label{lem:reverse}
If a set $A$ shatters a set $B$ with $|B|=2^n$, then $B$ shatters a subset $A^*$ of $A$ with $|A^*|=n$.
\end{lemma}

\begin{proof}
Without loss of generality we assume that $B$ is the set of all binary sequences of length $n$. 
Then every vertex $a\in A$ defines a Boolean function of $n$ variables (the neighbourhood of $a$ consists of the binary sequences, where the function takes value 1). 
For each $i=1,\ldots,n$, let us denote by $a_i$ the Boolean function such that $a_i(x_1,\ldots,x_n)=1$ if and only if $x_i=1$. 
Let $A'$ be an arbitrary subset of $A^*=\{a_1,\ldots,a_n\}$ and $\alpha=(\alpha_1,\ldots,\alpha_n)$ its characteristic vector, 
i.e. $\alpha_i=1$ if and only if $a_i\in A'$. Clearly, $\alpha\in B$ and $N(\alpha)\cap A^*=A'$.
Therefore, $B$ shatters $A^*$. 
\end{proof}

\begin{lemma}\label{lem:k(n)}
For every $n$, there exists a $k=k(n)$ such that every graph $G$ with $vc(G)=k$ contains one of 
$W_n,\overline{W}_n, D_n, \overline{D}_n$ as an induced subgraph. 
\end{lemma}

\begin{proof}
Define $k=R(2^{R(n)})$, where $R$ is the symmetric Ramsey number. Since $vc(G)=k$, there are two subsets $A$ and $B$ of $V(G)$ such that $|A|=k$ 
and $B$ shatters $A$. By definition of $k$, $A$ must have a subset $A'$ of size $2^{R(n)}$ which is a clique or an independent set.
Clearly, $B$ shatters $A'$ and hence, by Lemma~\ref{lem:reverse}, $A'$ shatters a subset $B'$ of $B$ of size $R(n)$. 
Then $B'$ must have a subset $B''$ of size $n$ which is either a clique or an independent set. Now $G[A'\cup B'']$ 
is either bipartite or co-bipartite or split graph, $|B''|=n$ and $A'$ shatters $B''$. Therefore, $G[A'\cup B'']$
contains one of $W_n,\overline{W}_n, D_n, \overline{D}_n$ as an induced subgraph.   
\end{proof}

\begin{theorem}\label{thm:VC}
The classes of bipartite, co-bipartite and split graphs are the only three minimal hereditary classes of graphs 
of unbounded VC-dimension. 
\end{theorem}

\begin{proof}
Clearly these three classes have unbounded VC-dimension, since they contain $W_n,\overline{W}_n, D_n, \overline{D}_n$
with arbitrarily large values of $n$.

Now let $X$ be a hereditary class containing none of these three classes. Therefore, there is a bipartite graph $G_1$,
a co-bipartite graph $G_2$ and a split graph $G_3$ which are forbidden for $X$. Denote by $n$ the maximum number of 
vertices in these graphs. 

Assume that VC-dimension is not bounded for graphs in $X$ and let $G\in X$ be a graph with $vc(G)=k$, where $k=k(n)$ is from Lemma~\ref{lem:k(n)}.
Then $G$ contains one of $W_n,\overline{W}_n, D_n, \overline{D}_n$, say $W_n$. Since $W_n$ is $n$-universal (Lemma~\ref{lem:universal}),
it contains $G_1$ as an induced subgraph, which is impossible because $G_1$ is forbidden for graphs in $X$. This contradiction
shows that VC-dimension is bounded in the class $X$. 
\end{proof}


\section{The speed of hereditary properties}
\label{sec:speed}

Theorem~\ref{thm:VC} together with Alekseev-Bollob\'as-Thomason Theorem imply the following conclusion.

\begin{theorem}
The entropy of a hereditary property $\cal P$ equals zero if and only if the VC-dimension of graphs in $\cal P$ is bounded by a constant.
\end{theorem}  

In the rest of the section we characterize three other jumps of the speed by means of graph parameters.
In the proofs we use the results and notation of Sections~\ref{sec:parameters}.
 
\subsection{Hereditary classes of constant speed}

The family of hereditary classes of constant speed constitute the lowest
layer of the lattice of hereditary properties. To characterize classes in this layer,
we introduce one more parameter as follows. Let $S(G)$ and $s(G)$ denote be the size of a largest and a smallest similarity class in $G$, respectively.
Then the {\it similarity difference} of $G$ is $S(G)-s(G)$.

\begin{theorem}\label{thm:constant}
The speed of a hereditary property $\cal P$ is constant if and only if the similarity difference of graphs in $\cal P$ is bounded by a constant.
\end{theorem}
\begin{proof}
To prove the theorem we introduce the following classes of graphs:
\begin{itemize}
\item[$\cal R$] the class of graphs each of which is either an edgeless graph or a star,
\item[${\cal E}^1$] the class of graphs with at most one edge.
\end{itemize}
The proof of the theorem is based on the following claim.
\begin{itemize}
\item[(*)] {\it If none of $\cal R$, ${\cal E}^1$, $\cal \overline{R}$, ${\cal \overline{E}}^1$ is a subclass of $\cal P$, then  
$\cal P$ contains finitely many graphs different from complete and empty graphs}.

\medskip
Indeed, if none of the four classes is a subclass of $\cal P$, then there is a number $p$ such that none of the following graphs belongs to $\cal P$:
$K_{1,p}$, $\overline{K}_{1,p}$, $H^1_p$, $\overline{H}^1_p$, where $H^1_p$ is a graph from ${\cal E}^1$ containing one edge and $p$ isolated vertices. 

Let $G$ be a graph in $\cal P$ which is neither complete nor edgeless. Then $G$ contains a vertex $x$ which has both a neighbour 
$y$ and a non-neighbour $z$. The remaining vertices of $G$ can be partitioned (with respect to $x,y,z$) into at most eight subsets. 
For $U\subseteq \{x,y,z\}$, we denote $V_U=\{v\not \in \{x,y,z\}\ : \ N(v)\cap \{x,y,z\}=U\}$. To prove the claim, let us show that 
each of $V_U$ contains at most $R(p)$ vertices, where $R(p)$ is the symmetric Ramsey number. 

If $U=\emptyset$, then $|V_U|<R(p)$  because a clique of size $p$ in $V_U$ together with $x$ create an induced $\overline{K}_{1,p}$,
while an independent set of size $p$ in $V_U$ together with $x$ and $y$ create an induced $H^1_p$, which is impossible because both graphs are forbidden in $X$.

If $|U|=1$, say $U=\{x\}$, then $|V_U|<R(p)$  because a clique of size $p$ in $V_U$ together with $y$ or $z$ create an induced $\overline{K}_{1,p}$,
while an independent set of size $p$ in $V_U$ together with $x$ create an induced $K_{1,p}$, which is impossible because both graphs are forbidden in $X$.

For $|U|>1$, the result follows by complementary arguments.
\end{itemize}

\noindent
Now we turn to the proof of the theorem and assume first that $\cal P$ has a constant speed, then none of $\cal R$, ${\cal E}^1$, $\cal \overline{R}$, ${\cal \overline{E}}^1$ is a subclass of $\cal P$,
since ${\cal R}_n={\cal \overline{R}}_n=n+1$ and ${\cal E}^1_n={\cal \overline{E}}^1_n={n\choose 2}+1$. Therefore, by Claim (*), 
$\cal P$ contains finitely many graphs different from complete and empty graphs, and hence the similarity difference of graphs in $\cal P$ is bounded by a constant.

Conversely, assume that the similarity difference of graphs in $\cal P$ is bounded by a constant. Then none of $\cal R$, ${\cal E}^1$, $\cal \overline{R}$, ${\cal \overline{E}}^1$ is a subclass of $\cal P$,
because this parameter is unbounded in each of the four listed classes. Therefore, by Claim (*), 
$\cal P$ contains finitely many graphs different from complete and empty graphs, implying that $\cal P$ has a constant speed.
\end{proof}

\subsection{Hereditary classes with a polynomial speed of growth}

According to the proof of Theorem~\ref{thm:constant}, the four minimal classes above the constant layer are 
$\cal R$, ${\cal E}^1$, $\cal \overline{R}$ and ${\cal \overline{E}}^1$. Each of them contains 
polynomially many $n$-vertex labelled graphs and hence the layer following the constant one is polynomial.

\begin{theorem}\label{thm:polynomial}
Let $\cal P$ be a hereditary class above the constant layer, then the speed of $\cal P$ is polynomial if and only if  the complex degree and $c$-matching number are bounded for graphs in $\cal P$ by a constant.
\end{theorem}

\begin{proof}
If the speed of a hereditary property $\cal P$ is polynomial, then it contains none of the classes ${\cal B,S,Q,M},\overline{\cal B},\overline{\cal S},\overline{\cal Q}, \overline{\cal M}$,
because ${\cal S}_n=2^n-n$, ${\cal B}_n=2^{n-1}$, ${\cal Q}_n=n2^{n-1}-n(n+1)/2+1$ and
${\cal M}_n$ is at least $\lfloor n/2\rfloor !$. This implies by Theorems~\ref{thm:complex-degree} and~\ref{thm:c-matching} that the complex degree and $c$-matching number are bounded for graphs in $\cal P$.

Conversely, assume that the complex degree and $c$-matching number are bounded for graphs in $\cal P$ by a constant $k$ and let 
$G$ be a graph in $\cal P$ with $n$ vertices. Without loss of generality, let $\mu(G)\le k$.
Therefore, $n-2k$ vertices form an independent set $I$ in $G$. Also, since the complex degree is bounded in $G$,  
every vertex $v$ outside of $I$ has at most $k$ neighbours or at most $k$ non-neighbours in $I$. By removing from $I$ either $k$ neighbours or $k$ non-neighbours of $v$,
for each $v\not \in I$, we transform $I$ into a similarity class of size $n-c$, where $c\le 2k(k+1)$. It is not difficult to see that the number of labelled graphs
on $n$ vertices with a similarity class of size $n-c$ is ${n\choose c}2^{{c+1\choose 2}+1}$ and hence the speed of $\cal P$ is polynomial.  
\end{proof}

\subsection{Hereditary classes with an exponential speed of growth}

The proof of theorem~\ref{thm:polynomial} tells us that there are eight minimal classes above the polynomial layer, 
namely ${\cal B,S,Q,M},\overline{\cal B},\overline{\cal S},\overline{\cal Q},$ and $\overline{\cal M}$. 
Each of these classes contains at least exponentially many $n$-vertex labelled graphs and hence the next layer is the exponential one.

\begin{theorem}\label{thm:exponential}
Let $\cal P$ be a hereditary class above the polynomial layer, then the speed of $\cal P$ is exponential if and only if the neighbourhood diversity is bounded for graphs in $\cal P$ by a constant.
\end{theorem}

\begin{proof}
If the speed of $\cal P$ is exponential, the $\cal P$ contains none of the following classes as a subclass: 
$${\cal M}, {\cal M}^{bc}, {\cal Z}, \overline{\cal M}, \overline{\cal M}^{bc}, \overline{\cal Z}, {\cal M}^*, \overline{\cal M}^*, {\cal Z}^*.$$
Indeed, as stated previously, the number of $n$-vertex labelled graphs in ${\cal M}$ is at least $\lfloor n/2\rfloor !$.
It is not difficult to see that the same is true for each of the other listed classes. This implies, by Theorem~\ref{thm:diversity}, 
that the neighbourhood diversity is bounded for graphs in $\cal P$.

Conversely, assume the neighbourhood diversity is bounded for graphs in $\cal P$ by a constant $k$. It is not difficult to see that
the number of labeled graphs on $n$ vertices 
with at most $k$ similarity classes is at most $k^n2^{\binom{k}{2}+k}$. Therefore, the speed of $\cal P$ is exponential.
\end{proof}

\section{Conclusion and open problems}
\label{sec:con}

In this paper, we revealed several relationships between the speed of a hereditary property and
boundedness of some graph parameters.
In particular, we have shown that a hereditary class $X$ has at most exponential speed of growth if and only if the neighbourhood diversity of graphs in $X$ is bounded by a constant.
Since all the minimal classes above the exponential layer have a factorial speed, the next layer is the factorial one. 

The structure of classes at the bottom of the factorial layer,
with speeds below the Bell number, have been characterized in \cite{Bell}. This is precisely the family of classes where a certain graph parameter, called in \cite{uni} uniformicity, is bounded by a constant. 
The set of all minimal classes of unbounded uniformicity have been recently described in \cite{SIDMA}. However, globally, i.e. beyond graphs of bounded uniformicity,  
the structure of classes in the factorial layer remains an open question. 
Also, it is not clear what are the parameters responsible for a factorial speed of growth. According to the results of the present paper, 
any such parameter must be a restriction of VC-dimension. On the other hand, it should generalize clique-width, since all classes of bounded clique-width have at most factorial 
speed of growth \cite{Allen}.

We conclude the paper by observing that there is an intriguing relationship between the three unavoidable classes of bipartite graphs of unbounded neighbourhood diversity 
($\cal M$, ${\cal M}^{bc}$ and $\cal Z$, see Lemma~\ref{lem:nd}) and 
the three unavoidable structures in the Canonical Ramsey Theorem \cite{canonical} in case of coloring 2-subsets. In this case, the theorem can be stated as follows:
for every positive integer $\ell$, there exists a positive integer $n=n(\ell)$ such that if the 2-subsets of an $n$-set
are colored with arbitrarily many colors, then the set contains a subset of size $\ell$ which is either 
\begin{itemize}
\item monochromatic, i.e. all the 2-subsets have the same color, or
\item rainbow, i.e. the 2-subsets have pairwise different colors, or 
\item skew, i.e. the elements of the subset can be ordered $x_1,x_2,\ldots,x_{\ell}$
in such a way that the subsets $\{x_i,x_j\}$ $(i<j)$ and $\{x_p,x_t\}$ ($p<t$) have the same color if and only if $i=p$. 
\end{itemize}
In this statement, monochromatic 
and rainbow colorings are complement to each other, similarly to classes $\cal M$ and ${\cal M}^{bc}$, while a skew coloring is self-complementary, as in the case of chain graphs (the class $\cal Z$). 
We ask whether the two results (the existence of three unavoidable classes of bipartite graphs of unbounded neighbourhood diversity and  three unavoidable colorings in the Canonical Ramsey Theorem)
can be derived from each other.

\end{document}